\documentclass{article}

\usepackage[margin=3.5cm]{geometry}
\usepackage{amsfonts}
\usepackage{mathrsfs}
\usepackage{amssymb}
\usepackage{latexsym}
\usepackage{graphicx}
\usepackage{enumerate}
\usepackage{amsmath}
\usepackage{mathtools}
\usepackage{color}
\usepackage{float}
\usepackage{array}

\usepackage{wasysym}
\usepackage{pifont}

\usepackage{centernot}

\usepackage{eucal}

\mathtoolsset{centercolon}

\usepackage{amsthm}
\newtheorem{theorem}{Theorem}[section]

\newtheorem{proposition}[theorem]{Proposition}

\usepackage{amssymb,tikz}

\newcommand{\mysetminusD}{\hbox{\tikz{\draw[line width=0.6pt,line cap=round] (3pt,0) -- (0,6pt);}}}
\newcommand{\mysetminusT}{\mysetminusD}
\newcommand{\mysetminusS}{\hbox{\tikz{\draw[line width=0.45pt,line cap=round] (2pt,0) -- (0,4pt);}}}
\newcommand{\mysetminusSS}{\hbox{\tikz{\draw[line width=0.4pt,line cap=round] (1.5pt,0) -- (0,3pt);}}}

\newcommand{\mysetminus}{\mathbin{\mathchoice{\mysetminusD}{\mysetminusT}{\mysetminusS}{\mysetminusSS}}}

\renewcommand{\phi}{\varphi}

\newcommand{\defin}[1]{\textbf{#1}}

\newcommand{\lthen}{\rightarrow}

\newcommand{\val}[1]{[\![ #1 ]\!]}

\newcommand{\commentout}[1]{}

\usepackage{booktabs, tabularx}

\renewcommand{\L}{\mathcal{L}}

\renewcommand{\implies}{\Rightarrow}
\newcommand{\p}{\bar{p}}
\newcommand{\Bel}{\mathsf{Bel}}
\newcommand{\Belp}{\Bel^{\pi}}

\title{Uncertain Interpretation}
\author{Adam Bjorndahl}

\begin{document}


\begin{center}
{\LARGE A Logic of Uncertain Interpretation}\\
\vspace{5mm}
{\large Adam Bjorndahl}\\
Carnegie Mellon University
\end{center}

\begin{abstract}
We introduce a logical framework for reasoning about ``uncertain interpretations'' and investigate two key applications: a new semantics for implication capturing a kind of ``meaning entailment'', and a conservative notion of ``evidentially supported'' belief that takes the form of a Dempster-Shafer belief function.
\end{abstract}

\section{Introduction}

We do not always know how to interpret the statements that we hear, the observations that we make, or the evidence that we gather. Traditional frameworks for reasoning about uncertainty and belief revision typically suppose that new information is presented definitively: there is no question about \textit{what} was learned. The paradigm of \emph{Bayesian conditioning} exemplifies this assumption: ``evidence'' takes the simple form of an event $E$, and belief revision proceeds by updating probabilities accordingly: $\pi \mapsto \pi(\cdot \, | \, E)$.

In order to capture the kind of uncertainty about interpretation we wish to reason about, we change the fundamental representation of events so that the sets they correspond to are themselves variable---the ``true meaning'' of a statement thus becomes itself an object of uncertainty. This approach follows in the spirit of other recent work \cite{BO19,B20}, expanding on it along two key dimensions.

First, the scope of meaning variability is broadened to construct an entire logic of statements of uncertain interpretation. This paves the way for a new semantics of \textit{implication} capturing a kind of ``meaning entailment'' we do not get from the standard material conditional. The second and more major expansion incorporates probabilistic representations of uncertainty. In this context we show how variable interpretations of evidence together with our new type of implication lead naturally to a conservative notion of ``evidentially supported'' belief that takes the form of a Dempster-Shafer belief function. We conclude with some remarks about evidence combination in this framework that suggest promising directions for continuing research.

\section{Foundations}

We begin, as usual, with a nonempty set $X$ of \emph{states} or \emph{possible worlds}, intuitively representing all the different ways things might be. An \emph{event} is a subset $E \subseteq X$, corresponding to the assertion that the world is \textit{that} way (i.e., it is one of the worlds in $E$). Accordingly, it is natural (and standard) to interpret a statement $\phi$ (drawn from some formal language) as corresponding to an event $\val{\phi} \subseteq X$, namely, the event that $\phi$ is true.

The language of classical propositional logic furnishes a simple and familiar example. Let $\textsc{prop}$ be a countable collection of \emph{primitive propositions}, and let $\L$ denote the basic propositional language given by
$$\phi ::= p \, | \, \lnot \phi \, | \, \phi \land \psi,$$
where $p \in \textsc{prop}$ (with the other Boolean connectives treated as abbreviations in the usual way). Then, given a nonempty set $X$ and a \emph{valuation} $v \colon \textsc{prop} \to 2^{X}$, we recursively extend this valuation to all formulas in $\L$ via
\begin{eqnarray*}
\val{p} & = & v(p)\\
\val{\lnot \phi} & = & X \mysetminus \val{\phi}\\
\val{\phi \land \psi} & = & \val{\phi} \cap \val{\psi}.
\end{eqnarray*}
Thus the primitive propositions are interpreted directly as events (via the valuation), and this is extended to all Boolean connectives in the standard way, interpreting negation as complementation and conjunction as intersection. Of course, the possible worlds framework is intended to represent not just truth and falsity but also belief and uncertainty---for example, by assigning probabilities to events in $X$---and we will get to this business in Section \ref{sec:bel}. But first we introduce the twist that is the basis for the present work.

\subsection{Variable valuations} \label{sec:var}

We generalize the interpretation of primitive propositions by defining a \defin{variable valuation} $V$ to be a function on \textsc{prop} that outputs not subsets of $X$ but rather functions from $X$ to subsets of $X$. Formally, for each $p \in \textsc{prop}$, we have $V(p) \colon X \to 2^{X}$. To reduce notational clutter, we omit $V$ altogether whenever it is clear from context, writing $p(x)$ instead of $V(p)(x)$.

The fundamental aim here is to capture \textit{uncertainty about the meaning of propositions}; moreover, we wish to encode this uncertainty using the same formal mechanism that encodes all uncertainty, namely, variation across the state space $X$. For instance, at $x \in X$ we can say that the ``true meaning'' or ``correct interpretation'' of $p$ is the event $p(x) \subseteq X$, whereas at $y \in X$ the true meaning of $p$ is instead $p(y) \subseteq X$, and so on. As a consequence, an agent who considers both $x$ and $y$ possible (e.g., assigning to each a positive probability) would correspondingly consider both $p(x)$ and $p(y)$ to be possible interpretations of the (syntactic) statement $p$.

Naturally, we recover the standard notion of a valuation for a primitive proposition $p \in \textsc{prop}$ whenever the function $V(p)$ is constant, say $p(x) = U$ for all $x \in X$ (hencforth we denote this by $p(x) \equiv U$). In this case, intuitively, there is only one correct interpretation of $p$, namely $U$. So the present framework subsumes the standard one in the obvious way---identifying classical events (i.e., subsets of the state space) with constant functions.

While of course we do not wish to restrict attention to constant functions, we will often find it convenient to impose a weaker assumption. To explain this assumption it is helpful to first interrogate the concept of \textit{truth} in a framework with variable valuations: at which worlds should we say that $p$ is actually true? The intuition here is relatively straightforward: $p$ is true at $x$ just in case the correct interpretation of $p$ at $x$, namely $p(x)$, includes $x$ itself. Thus we define the \defin{truth set} of $p$ as follows:
$$\val{p} = \{x \in X \: : \: x \in p(x)\}.$$
Note that if $p(x) \equiv U$ then $\val{p} = U$.

Now we state the aforementioned assumption: say that $p \colon X \to 2^{X}$ is \defin{coherent} if, for all $x,y \in X$, $y \in p(x)$ implies $y \in p(y)$. Equivalently, for all $x \in X$, $p(x) \subseteq \val{p}$. Intuitively, this says that at all worlds $x$, the correct interpretation of $p$ entails that $p$ is true. It is easy to see that every constant function is coherent, but not conversely. Note also that when $p$ is coherent,
$$\val{p} = \bigcup_{x \in X} p(x),$$
that is, the truth set of $p$ is just the set of all worlds consistent with \textit{some} interpretation of $p$.

It is useful to observe that we can transform any incoherent variable  valuation into a strictly stronger, coherent one with the same truth set.

\begin{proposition}\label{pro:coh}
Let $p \colon X \to 2^{X}$, and define $\bar{p} \colon X \to 2^{X}$ by setting $\bar{p}(x) = p(x) \cap \val{p}$ for all $x \in X$. Then $\bar{p}$ is coherent and $\val{p} = \val{\bar{p}}$.
\end{proposition}

\begin{proof}
We begin with the latter claim. First suppose that $x \in \val{p}$; this means that $x \in p(x)$. Then clearly $x \in p(x) \cap \val{p} = \bar{p}(x)$, so $x \in \val{\bar{p}}$. Conversely, suppose that $x \in \val{\bar{p}}$; then $x \in \bar{p}(x) \subseteq \val{p}$.

Having established this, we can see directly that $\bar{p}(x) \subseteq \val{p} = \val{\bar{p}}$, hence $\bar{p}$ is coherent.
\end{proof}

A variable valuation can be extended to all formulas in $\L$ in a straightforward, pointwise manner. After all, if the correct interpretation of $p$ at $x$ is $p(x)$, then presumably the correct interpretation of $\lnot p$ at $x$ is $X \mysetminus p(x)$; similarly, if the correct interpretations of $p$ and $q$ at $x$ are $p(x)$ and $q(x)$, the the correct interpretation of $p \land q$ at $x$ would be $p(x) \cap q(x)$. Thus we recursively define
\begin{eqnarray*}
p(x) & = & V(p)(x)\\
(\lnot \phi)(x) & = & X \mysetminus \phi(x)\\
(\phi \land \psi)(x) & = & \phi(x) \cap \psi(x),
\end{eqnarray*}
and extending the definition above, for all $\phi \in \L$, we define the \defin{truth set of $\phi$} to be:
$$\val{\phi} = \{x \in X \: : \: x \in \phi(x)\}.$$
The following is then an easy exercise.

\begin{proposition}
For all $\phi, \psi \in \L$, we have $\val{\lnot \phi} = X \mysetminus \val{\phi}$ and $\val{\phi \land \psi} = \val{\phi} \cap \val{\psi}$.
\end{proposition}

\subsection{A running example} \label{sec:exa}

A simple, running example will be useful to illustrate new concepts as we introduce them and facilitate the comparison with standard models. Suppose that your friend has flipped a fair coin that landed either heads ($H$) or tails ($T$). They can see the coin but you cannot; you await a report from them as to which way the coin landed. However, you are unsure about your friend's disposition in making this report: perhaps they will accurately tell you the way the coin landed ($accurate$), but you also consider it possible that they are bored or distracted and instead will either always say heads ($sayheads$) or always say tails ($saytails$), regardless of the true state of the coin. Naturally, we can represent the possibilities at play here using six states:
$$X = \{H,T\} \times \{accurate, sayheads, saytails\}.$$

Suppose now that your friend says ``Heads''; call this statement $p$. It is easy to see that $p$ is compatible with exactly three of the six states above, namely, the event
$$P = \{(H, accurate), (H, sayheads), (T, sayheads)\};$$
indeed, these are precisely the states in which your friend would say ``Heads''. So the standard way of associating an event to the proposition $p$ would be to set $v(p) = P$.

However, the nature of the example makes it clear that the actual meaning of $p$ is something that depends on your friend's disposition, that is, the true meaning of their utterance is something you are uncertain about. Roughly speaking, we might say that if your friend is accurate then $p$ means the coin landed heads, if your friend always says heads then $p$ doesn't tell you anything, and if your friend always says tails then $p$ is not possible at all. So, shifting to the variable valuation framework, as a first pass we might define
\begin{displaymath}
p(x,y) =
\begin{cases}
\{(H, accurate), (H, sayheads), (H, saytails)\} & \textrm{if $y = accurate$}\\
X & \textrm{if $y = sayheads$}\\
\emptyset & \textrm{if $y = saytails$}.
\end{cases}
\end{displaymath}

This is certainly more complicated than the event $P$! We will leverage this additional structure subsequently, but for the moment let us begin by observing that we can recover $P$ from this variable interpretation of $p$ quite directly: one easily checks that for all $(x,y) \in X$, $(x,y) \in p(x,y)$ iff $(x,y) \in P$, so by definition $\val{p} = P$.

On the other hand, $p$ is not coherent according to our definition. For example, $(H, saytails) \in p(H, accurate)$ even though $(H, saytails) \notin \val{p}$. And there is something a bit odd about this: intuitively, whatever the statement ``Heads'' might mean, it \textit{at least} rules out that ``Tails'' is what was said. Proposition \ref{pro:coh} gives us a natural way to make $p$ coherent without changing its truth set:
\begin{displaymath}
\bar{p}(x,y) =
\begin{cases}
\{(H, accurate), (H, sayheads)\} & \textrm{if $y = accurate$}\\
P & \textrm{if $y = sayheads$}\\
\emptyset & \textrm{if $y = saytails$},
\end{cases}
\end{displaymath}
We take this variable valuation to provide a better model of the statement in question.

\subsection{Implication through meanings} \label{sec:imp}

According to our definitions, the Boolean implication symbol $\lthen$ is simply an abbreviation; in particular, $\phi \lthen \psi$ is shorthand for $\lnot \phi \lor \psi$, and so
$$(\phi \lthen \psi)(x) = (X \mysetminus \phi(x)) \cup \psi(x).$$
This seems a perfectly respectable, pointwise rendering of the material conditional, and indeed one easily checks that $\phi \lthen \psi$ is true at $x$ just in case either $\phi$ is false at $x$ or $\psi$ is true at $x$:
$$\val{\phi \lthen \psi} = (X \mysetminus \val{\phi}) \cup \val{\psi}.$$
Notably, however, the present framework offers a very different but also quite natural way of formalizing what ``$\phi$ implies $\psi$'' could mean: we might say that $\phi$ implies $\psi$ at $x$ just in case the correct interpretation of $\phi$ at $x$ \textit{entails} the correct interpretation of $\psi$ at $x$. To capture this formally we introduce a new symbol, $\implies$, and for $\phi,\psi \in \L$ define
$$\val{\phi \implies \psi} = \{x \in X \: : \: \phi(x) \subseteq \psi(x)\}.$$
Note that here we have defined the truth set of $\phi \implies \psi$ directly, rather than associating this new expression with a function from $X$ to $2^{X}$ as we did with the formulas of $\L$.\footnote{Intuitively, although there may be uncertainy about whether or not $\phi \implies \psi$ is \textit{true}, there is no question as to what it \textit{means}; thus one could represent it with the constant function that always returns the set above.}

A concrete example will be helpful here, for which we return to our coinflip scenario. Suppose that $h$ represents the (unambiguous) statement that the coin actually landed heads---in particular, let the variable valuation for $h$ be the constant function that always returns the set $\{(H, accurate), (H, sayheads), (H, saytails)\}$. Recall that $\p$, by contrast, is the (coherent) statement ``Heads'' as made by your friend, the meaning of which is uncertain in the sense described above.

We compare the truth set of $\p \lthen h$ with that of $\p \implies h$. For the former, we have
\begin{eqnarray*}
\val{\p \lthen h} & = & (X \mysetminus \val{\p}) \cup \val{h}\\
& = & X \mysetminus \{(T, sayheads)\},
\end{eqnarray*}
since $(T, sayheads)$ is the only state where $\p$ is true and $h$ is false. On the other hand, referring to the variable valuation of $\p$, we can see that $\p(x,y) \subseteq \{(H, accurate), (H, sayheads), (H, saytails)\}$ iff $y \neq sayheads$. This corresponds to the fact that it is precisely when your friend always says ``Heads'' that their statement \textit{doesn't} entail that the coin actually landed heads. This intuitive difference is perhaps somewhat clearer if we focus on those worlds in which $\p$ actually holds:
$$\val{\p} \cap \val{\p \lthen h} = \{(H, accurate), (H, sayheads)\},$$
whereas
$$\val{\p} \cap \val{\p \implies h} = \{(H, accurate)\}.$$
At the world $(H, sayheads)$, $\p \lthen h$ is true simply because $\p$ and $h$ are both true (as usual for a material conditional). By contrast, $\p \implies h$ is false at this world because the statement $\p$ does not in itself guarantee the truth of $h$ (even though $h$ happens to be true)---in this world your friend just says ``Heads'' no matter what! At the world $(H, accurate)$ the two implication are both true, but for different reasons: $\p \lthen h$ is true, as above, because both $\p$ and $h$ are true, whereas $\p \implies h$ is true because at this world your friend is accurate and therefore their statement actually guarantees the truth of $h$.

\section{Belief} \label{sec:bel}

At a high level, we might picture an agent's ``attitude toward uncertainty'' as stemming from two basic components: first, their subjective, \emph{prior beliefs} about how the world is, given by a probability measure $\pi$ on $X$; and second, some body of \emph{total evidence} they have about the world which can be used to constrain their prior beliefs. Classically, the ``evidence'' in this second component is simply some (true) event $E \subseteq X$, and it ``constrains'' prior beliefs through Bayesian conditioning, producing the posterior beliefs $\pi(\cdot \, | \, E)$.

In the present framework we generalize our understanding of ``events'', and therefore of ``evidence'', to include variable valuations, thus replacing $E \subseteq X$ with $\psi \in \L$ in the sketch above. As we will show, the sense in which such a $\psi$ might ``constrain'' prior beliefs $\pi$ is richer than merely conditioning.

\subsection{Probability and conditioning}

A standard way of representing belief is to attach a probability measure $\pi$ to the state space $X$. In order to streamline the presentation of the novel aspects of this model, we take $X$ to be finite. For any event $A \subseteq X$, the quantity $\pi(A)$ intuitively represents the agent's degree of belief that $A$ is true. Extending this notion to formulas with variable valuations is straightforward given our definition of truth sets: for $\phi \in \L$, the agent's degree of belief in $\phi$ is simply the probability of the \textit{event that $\phi$ is true}, namely, $\pi(\val{\phi})$. When $\pi(\val{\phi}) = 1$, we say that the agent \emph{believes $\phi$}.

It is likewise standard to represent belief revision via Bayesian conditioning; specifically, if the agent learns that some event $B$ is the case, then their new degree of belief in any event $A$ is the conditional probability $\pi(A \, | \, B)$ (provided $\pi(B) \neq 0$, otherwise this is undefined). As above, a straightforward generalization to the variable valuation framework is readily available: intuitively, if we interpret ``learning $\psi$'' to mean ``learning that $\psi$ is true'', it is natural to define the agent's degree of belief in $\phi$ given $\psi$ to be $\pi(\val{\phi} \, | \, \val{\psi})$.%
\footnote{Notice that under these definitions, the variable nature of $\phi$ and $\psi$ does not come to the fore, since all the formulas involve only the corresponding truth sets $\val{\phi}$ and $\val{\psi}$ (which are events). One might, instead, propose a more ``pointwise'' notion of conditioning, for example by defining
$$\pi(\phi \, | \, \psi) = \sum_{E \subseteq X} \pi(\val{\phi} \, | \, E) \cdot \pi(\psi^{-1}(E)).$$
Loosely speaking, this formula says that to update on $\psi$ is to take a weighted average over all possible updating events $E$, weighted by the (prior) probability that $E$ is the correct interpretation of $\psi$; in other words (speaking even more loosely), to update on the ``expected interpretation'' of $\psi$. While the intuition behind a definition like this has a certain appeal, it also comes with a host of interpretational difficulties, so it is not the approach we will pursue in this paper.}
When $\pi(\val{\phi} \, | \, \val{\psi}) = 1$, we say that the agent \emph{believes $\phi$ given $\psi$}.

We return once again to our running example to illustrate these basic points. First let us imagine that you rather suspect your friend is accurate (ascribing this event, say, $60\%$ likelihood), and otherwise find the other two possibilities as to their disposition equally likely; let's also assume that you think the coin is fair. These initial beliefs are captured with the prior $\pi$ given in the following table:
\begin{table}[H]
\begin{center}
\begin{tabular}{c|c|c}
 & $H$ & $T$\\
 \hline
$accurate$ & $3/10$ & $3/10$\\
$sayheads$ & $1/10$ & $1/10$\\
$saytails$ & $1/10$ & $1/10$\\
\end{tabular}
\end{center}
\caption{The prior $\pi$} \label{tbl:con}
\end{table}%
\noindent After hearing your friend say ``Heads''---that is, upon learning $\p$, which we recall has truth set
$$\val{\p} = \{(H, accurate), (H, sayheads), (T, sayheads)\}$$
---you can revise your beliefs on the basis of this new evidence:
\begin{table}[H]
\begin{center}
\begin{tabular}{c|c|c}
 & $H$ & $T$\\
 \hline
$accurate$ & $3/5$ & $0$\\
$sayheads$ & $1/5$ & $1/5$\\
$saytails$ & $0$ & $0$\\
\end{tabular}
\end{center}
\caption{The posterior $\pi(\cdot \, | \, \val{\p})$} \label{tbl:con2}
\end{table}%

For another simple example, suppose instead that you do not suspect your friend is anything but accurate; say your prior is given by $\pi'$:
\begin{table}[H]
\begin{center}
\begin{tabular}{c|c|c}
 & $H$ & $T$\\
 \hline
$accurate$ & $1/2$ & $1/2$\\
$sayheads$ & $0$ & $0$\\
$saytails$ & $0$ & $0$\\
\end{tabular}
\end{center}
\caption{The prior $\pi'$} \label{tbl:acc}
\end{table}%
\noindent This time, upon hearing ``Heads'', you come to fully believe that the coin indeed landed heads---even though, of course, this is not entailed by $\p$ (since $(T, sayheads) \in \val{\p}$).
\begin{table}[H]
\begin{center}
\begin{tabular}{c|c|c}
 & $H$ & $T$\\
 \hline
$accurate$ & $1$ & $0$\\
$sayheads$ & $0$ & $0$\\
$saytails$ & $0$ & $0$\\
\end{tabular}
\end{center}
\caption{The posterior $\pi'(\cdot \, | \, \val{\p})$} \label{tbl:acc2}
\end{table}%

Notice that in each of these examples, the fact that $\p$ has a variable valuation plays no role in the calculations---we are simply updating on the event $\val{\p} = P$ as we would do in the classical case to obtain a posterior probability. There is, however, a different ``attitude toward uncertainty'' we can examine in this context: roughly speaking, in addition to the probability of an event \textit{conditional} on the evidence, we can ask about the conditional probability that the evidence \textit{implies} the event in question.

\subsection{Evidential support}

Fix a prior $\pi$, a body of evidence represented by some $\psi \in \L$, and an event $A \subseteq X$ that we wish to assess in light of the evidence. As previously noted, we can think of an event $A$ as the truth set of some primitive proposition $a$ whose associated variable valuation is the constant function $a(x) \equiv A$ (so $\val{a} = A$).

We have already seen one way of making such an assessment, namely, by computing the posterior probability of $A$ given $\psi$: $\pi(\val{a} \, | \, \val{\psi})$. What about the probability that $\psi$ implies $a$? If we interpret ``implies'' using the material conditional, this is not a terribly interesting prospect, for we have
$$\pi(\val{\psi \lthen a} \, | \, \val{\psi}) = \pi\big((X \mysetminus \val{\psi}) \cup \val{a} \, | \, \val{\psi}\big) = \pi(\val{a} \, | \, \val{\psi}).$$
However, if we appeal instead to our alternative notion of implication (Section \ref{sec:imp}), we obtain something quite different. In general, for any $A \subseteq X$, let us define
$$\Belp_{\psi}(A) = \pi(\val{\psi \implies a} \, | \, \val{\psi}).$$
Thus $\Belp_{\psi}(A)$ is the probability assigned (by $\pi$) to the proposition that the evidence ($\psi$) actually entails $A$, conditional on the evidence being true. For this reason we might think of $\Belp_{\psi}$ as capturing something like ``evidentially supported/justified'' belief.

Unlike with the material conditional, this is not necessarily equal to the conditional probability of $A$ given $\psi$. One easy way to see this is by considering the special case where the evidence $\psi$ is also constant, say $\psi(x) \equiv E$. In this case, by definition, we have
$$\val{\psi \implies a} = \{x \in X \: : \: \psi(x) \subseteq a(x)\} = \{x \in X \: : \: E \subseteq A\} =
\begin{cases}
X & \textrm{if $E \subseteq A$}\\
\emptyset & \textrm{otherwise},
\end{cases}
$$
so (writing $\Belp_{E}$ for $\Belp_{\psi}$ to make the constancy of $\psi$ salient) we have
$$
\Belp_{E}(A) = \begin{cases}
1 & \textrm{if $E \subseteq A$}\\
0 & \textrm{otherwise}.
\end{cases}
$$
So in this special case, $\Belp_{E}$ ignores the prior $\pi$ completely and essentially just checks whether or not the evidence entails the event in question.\footnote{An attitude that some have called ``knowledge''!}

When $\psi$ has a variable interpretation, $\Belp_{\psi}$ may take on values between $0$ and $1$, as we shall see shortly. It is not hard to show that $\Belp_{\psi}$ has some reasonable (and perhaps familiar) properties.
\begin{proposition} \label{pro:bel}
For all $A,B \subseteq X$, we have:
\begin{enumerate}[(a)]
\item
$\Belp_{\psi}(A) \leq \pi(A \, | \, \val{\psi})$.
\item
$\Belp_{\psi}(X) = \Belp_{\psi}(\val{\psi}) = 1$ and $\Belp_{\psi}(\emptyset) = 0$.
\item
If $A \subseteq B$ then $\Belp_{\psi}(A) \leq \Belp_{\psi}(B)$.
\item
If $A$ and $B$ are disjoint, then $\Belp_{\psi}(A \cup B) \geq \Belp_{\psi}(A) + \Belp_{\psi}(B)$.
\end{enumerate}
\end{proposition}

Let's return one more time to our running example, and first focus on the prior probability $\pi$ given in Table \ref{tbl:con}. Recall that $h$ is the constant function that always returns the set $\{(H, accurate), (H, sayheads), (H, saytails)\}$, that is, the event that the coin landed heads. We can see in Table \ref{tbl:con} that $\pi(\val{h}) = 1/2$, and in Table \ref{tbl:con2} that $\pi(\val{h} \, | \, \val{\p}) = 4/5$, the latter corresponding the fact that once your friend says ``Heads'', it becomes much more likely (though not certain) that the coin actually landed heads. We can also compute
$$\Belp_{\p}(\val{h}) = \pi(\val{\p \implies h} \, | \, \val{\p}) = \frac{\pi(\val{\p \implies h} \cap \val{\p})}{\pi(\val{\p})} = \frac{\pi(\{(H, accurate)\})}{\pi(\val{\p})} = 3/5.$$
Why is this lower than the value of $\pi(\val{h} \, | \, \val{\p})$? Intuitively, $\Belp_{\p}(\val{h})$ is not the likelihood that the coin merely \textit{landed} heads given your information; instead, it is the likelihood that \textit{your information entails} that the coin landed heads (conditional on having received it). Since this latter is a stronger statement, its probability may be lower. In this particular case, of the 3 worlds you consider possible once you hear your friend say ``Heads'', two of them are worlds where the coin actually landed heads, but only one of those (namely, $(H, accurate)$) is a world where what you heard actually supports that conclusion.

It is also instructive to assess the event that your friend is accurate using these two methods. Let $a(x) \equiv \{(H, accurate), (T, accurate)\}$; then it is easy to check that we have $\pi(\val{a}) = 3/5$ and $\pi(\val{a} \, | \, \val{\p}) = 3/5$, whereas
$$\Belp_{\p}(\val{a}) = \pi(\val{\p \implies a} \, | \, \val{\p}) = \frac{\pi(\val{\p \implies a} \cap \val{\p})}{\pi(\val{\p})} = \frac{\pi(\emptyset)}{\pi(\val{\p})} = 0.$$
This is because, intuitively, no world compatible with $\p$ is a world in which $\p$ entails that your friend is accurate.

Switching to the scenario captured by the prior probability $\pi'$ given in Table \ref{tbl:acc}, we find some similarities and some differences. Recall that in this case you begin with subjective certainty that your friend is accurate. We can easily see in Table \ref{tbl:acc} that $\pi(\val{h}) = 1/2$, and in Table \ref{tbl:acc2} that $\pi(\val{h} \, | \, \val{\p}) = 1$. We also have
$$\Bel^{\pi'}_{\p}(\val{h}) = \pi'(\val{\p \implies h} \, | \, \val{\p}) = \frac{\pi'(\{(H, accurate)\})}{\pi'(\val{\p})} = 1,$$
since you ascribe probability $1$ not only to the coin having landed heads, but also to your friend being accurate, and thus to your evidence (i.e., their report) entailing that the coin landed heads. Incidentally, this shows that $\Belp_{\psi}$ is not \emph{factive} in the sense that it may assign the maximal value $1$ to an event that is strictly stronger than the conditioning event $\val{\psi}$---in particular, in this case, $(T, sayheads)$ is a world compatible with your evidence but at which $h$ is false despite being assigned $1$ by $\Bel^{\pi'}_{\p}$.

We also have $\pi'(\val{a}) = 1$ and $\pi(\val{a} \, | \, \val{\p}) = 1$, yet
$$\Bel^{\pi'}_{\p}(\val{a}) = \pi'(\val{\p \implies a} \, | \, \val{\p}) = \frac{\pi'(\val{\p \implies a} \cap \val{\p})}{\pi'(\val{\p})} = \frac{\pi'(\emptyset)}{\pi'(\val{\p})} = 0;$$
as before, this is because you have no evidence whatsoever that supports the claim that your friend is actually accurate. Thus we can see that not only can $\Belp_{\psi}$ attain its upper bound of $\pi(\cdot \, | \, \val{\psi})$ (per Proposition \ref{pro:bel}(a)), but it can also fall short of this upper bound in a maximal way (i.e., the difference can be $1$).

\subsection{Dempster-Shafer}

There is a tight connection between the Dempster-Shafer theory of belief functions and the notion of ``evidentially supported beliefs'' as defined in the previous section. Indeed, our notation was chosen to highlight precisely this relationship. We refer the reader to \cite{H03} for an excellent overview of Dempster-Shafer theory; for our present purposes we focus on one central construction.

A \defin{mass function} on $X$ is a map $m \colon 2^{X} \to [0,1]$ such that $m(\emptyset) = 0$ and $\sum_{A \subseteq X} m(A) = 1$ (i.e., it is a probability measure on the powerset of $X$). Roughly speaking, the mass of $A$, $m(A)$, is supposed to capture the degree of support the agent has for $A$ as opposed to any other subset of $X$ (including subsets or supersets of $A$). Given a mass function $m$ one defines a \defin{belief function} $\Bel_{m} \colon 2^{X} \to [0,1]$ by setting
$$\Bel_{m}(A) = \sum_{B \subseteq A} m(B);$$
in other words, $\Bel_{m}$ assigns to $A$ the sum of the masses of all the subsets of $A$. This is reminiscent of our definition of $\Belp_{\psi}(A)$, which assesses the probability that $\psi$ entails $A$ ranging over the various possible interpretations of $\psi$---each of which is a subset $\psi(x) \subseteq X$. Indeed, we might think of the ``mass'' of $\psi(x)$ as being given by the (conditional) probability that it is the \textit{correct} interpretation of $\psi$.

This can be made precise. Given a probability measure $\pi$ on a finite set $X$ and a formula $\psi \in \L$ with variable valuation $\psi \colon X \to 2^{X}$, define $m^{\pi}_{\psi}(A) = \pi(\psi^{-1}(A) \, | \, \val{\psi})$.

\begin{proposition}
$m^{\pi}_{\psi}$ is a mass function.
\end{proposition}

\begin{proof}
First observe that $m^{\pi}_{\psi}(\emptyset) = \pi(\psi^{-1}(\emptyset) \, | \, \val{\psi}) = 0$, since if $\psi(x) = \emptyset$ then $x \notin \val{\psi}$. Next we note that the collection $\{\psi^{-1}(A) \: : \: A \subseteq X\}$ is a partition of $X$, and thus
\begin{eqnarray*}
\sum_{A \subseteq X} m^{\pi}_{\psi}(A) & = & \sum_{A \subseteq X} \pi(\psi^{-1}(A) \, | \, \val{\psi})\\
& = & \pi\big(\bigcup_{A \subseteq X} \psi^{-1}(A) \, \big | \, \val{\psi}\big)\\
& = & \pi(X \, | \, \val{\psi})\\
& = & 1. \qedhere
\end{eqnarray*}
\end{proof}

\begin{proposition}
$\Bel_{m^{\pi}_{\psi}} = \Belp_{\psi}$.
\end{proposition}

\begin{proof}
For any $A \subseteq X$, we have
\begin{eqnarray*}
\Bel_{m^{\pi}_{\psi}}(A) & = & \sum_{B \subseteq A} m^{\pi}_{\psi}(A)\\
& = & \sum_{B \subseteq A} \pi(\psi^{-1}(B) \, | \, \val{\psi})\\
& = & \pi\big(\bigcup_{B \subseteq A} \psi^{-1}(B) \, \big | \, \val{\psi}\big)\\
& = & \pi(\{x \in X \: : \: \psi(x) \subseteq A\} \, | \, \val{\psi})\\
& = & \pi(\val{\psi \implies a} \, | \, \val{\psi})\\
& = & \Belp_{\psi}(A). \qedhere
\end{eqnarray*}
\end{proof}

Thus, variable valuations provide a framework in which Dempster-Shafer belief functions arise naturally from probability measures defined on the same state space and implement a comparatively more conservative (in the sense of Proposition \ref{pro:bel}(a)) notion of ``evidentially supported'' belief.

We conclude by noting that the additional structure of our setting also allows for a novel way of \textit{combining} belief functions---one not available in traditional Dempster-Shafer theory. Roughly speaking, this is because our mass functions $m^{\pi}_{\psi}$ are defined on subsets of $X$ which themselves are associated with points in $X$ (via $\psi$); this serves as a kind of background ``indexing'' that gives us a pointwise way of merging evidence that is simply absent in the standard framework.

\emph{Dempster's rule of combination} tells us how to combine two mass functions (and thus produce a new belief function): given $m_{1}$ and $m_{2}$, we set $(m_{1} \oplus m_{2})(\emptyset) = 0$ and for each $A \neq \emptyset$ define
$$(m_{1} \oplus m_{2})(A) = \frac{\displaystyle \sum_{A_{1} \cap A_{2} = A} m_{1}(A_{1}) \cdot m_{2}(A_{2})}{\displaystyle \sum_{B_{1} \cap B_{2} \neq \emptyset} m_{1}(B_{1}) \cdot m_{2}(B_{2})}.$$
Essentially, this definition just multiplies together the masses of any two sets that overlap, adds that product to the new mass of the intersection, and normalizes the result. Importantly, according to this formula, \textit{any} two sets that overlap contribute part of their mass to the new mass function.

By contrast, in our setting we can define the combination of $m^{\pi}_{\psi_{1}}$ and $m^{\pi}_{\psi_{2}}$ to be the mass function $m^{\pi}_{\psi_{1} \land \psi_{2}}$. Recall that $(\psi_{1} \cap \psi_{2})(x) = \psi_{1}(x) \cap \psi_{2}(x)$, so on this definition, the only intersections that count towards the new mass function are those of sets that are associated with the same underlying world $x \in X$. The investigation of this new paradigm for evidence combination and its broader import for epistemological questions is the subject of ongoing research.

\bibliographystyle{plain}

\end{document}